\documentclass{amsart}
\usepackage{amstext, amssymb, amsthm, amsfonts, latexsym,bbm,tikz}
\usepackage{mathtools}
\usetikzlibrary{matrix}
\usepackage{varioref}
\usepackage{float}
\usepackage{longtable}
\usepackage{todonotes}
\usepackage{bbm}
\usepackage[utf8]{inputenc}
\usepackage[english]{babel}
\usepackage{booktabs}
 
\usepackage{comment}
\usepackage{tikz-cd}
\usepackage{hyperref}
\usepackage[version=3]{mhchem}

\newtheorem{theorem}{Theorem}[section]

\theoremstyle{definition}

\theoremstyle{remark}
\newtheorem{remark}[theorem]{Remark}

\theoremstyle{Conjecture/open problem}

\theoremstyle{assumption}

\theoremstyle{conjecture}

\def\1{\mathbf{1}}

\def\cA{\mathcal{A}}

\def\cC{\mathcal{C}}

\def\cX{\mathcal{X}}

\def\R{\mathbb{R}}
\def\N{\mathbb{N}}

\numberwithin{equation}{section}

\begin{document}

\title{
Impossibility results for equating the Youden Index with average scoring rules and Tjur $R^2$-like metrics}

\author[Linard Hoessly]{Linard Hoessly}

\address{Data Center of the Swiss Transplant Cohort Study, University hospital Basel, Basel, 4031 Switzerland}
\email{linard.hoessly@hotmail.com}

\date{}


\date{}

\begin{abstract} 
We consider the Youden index fas well as measures evaluating predicted probabilities for the maximum-likelihood estimate of a logistic regression model with predictor the classifier. We give impossibility results showing that the Youden index can not equal any average of a real scoring rule nor any metric averaging over binary outcomes (0s and 1s) for any continuous real-valued scoring rule. This shows the obstractions of such potential equivalences and highlights the distinct roles these metrics play in diagnostic assessment.
 \end{abstract}

\maketitle
\color{black}
\section{Introduction}
Interest in classification is ubiquitious across fields like medicine, epidemiology, or machine learning \cite{hand_classification}. Given past data, the general goal is to find a rule to classify future examples. In medicine and epidemiology classification often involves distinguishing between disease and non-disease states. This can then be crucial for guiding individual treatments or informing public health actions. Classifiers are evaluated with metrics like the true positive rate (TPR), true negative rate (TNR), the Youden index, or other metrics, all aiming to evaluate aspects of diagnostic accuracy \cite{classification_powers,hand_classification}. In particular, the Youden index for a classifier is defined as 
\begin{equation}\label{Youden}
Youden=TPR+TNR-1
\end{equation}

On the other hand, binary classifiers are often built using logistic regression with a threshold probability. As there is no direct analog of an R-squared for logistic regression, various pseudo-R-squared metrics were introduced to judge the quality of the logistic model \cite{harrell2015regression}. Tjur's $R^2$ coefficient \cite{Tjur} measures a model's discriminatory power in logistic regression through the average difference in predicted probabilities between the binary outcomes, making it a popular tool to evaluate a models  predictive capacity \cite{harrell2015regression}. Another common tool in model evaluation assessing prediction models are scoring rules \cite{scoring_rule}, which compare probabilistic predictions to observed outcome. In practice, an average over scoring rules is considered. 

Let 
\begin{equation}\label{n}
[n]=\{1,\cdots,n\}
\end{equation} 
be the set of the first $n$ natural numbers.
To be more concrete, assume data $\cX=\{(x_i,y_i)| \ i\in [n]\}$ are given, where the $x_i$ are the predictors and $y_i$ the binary outcomes. After fitting a logistic regression model, we denote the predicted probability of observation $i$ being $1$ by $\hat{p_i}$. Furthermore let $\cX_1=\{i|(x_i,y_i)\in \cX, y_i=1\}$ and $\cX_0=\{i|(x_i,y_i)\in \cX, y_i=0\}$. Then the Tjur's $R^2$ is given by
\begin{equation}\label{Tjur}
R^2_{Tjur}((\hat{p_i},y_i)_{i\in[n]})=\frac{1}{|\cX_1|}\sum_{i\in\cX_1}\hat{p_i}-\frac{1}{|\cX_0|}\sum_{i\in\cX_0}\hat{p_i}.
\end{equation}
An equivalence between the Youden index and Tjur $R^2$ for a logistic regression with the classification as independent variable was claimed in \cite{HUGHES2017801} and recently shown not to hold \cite{HOESSLY202528}.
 Denote all possible distributions on the set $\{0,1\}$ by $\Delta_1$, i.e.,
$$
\Delta_1=\{(p_0,p_1)\in\R^2_{\geq 0}|p_0+p_1=1\}.
$$
A real-valued scoring rule for binary outcomes is a function that maps a probability on the binary set and a binary outcome to the reals, i.e.
\begin{equation}\label{score}
S(\cdot,\cdot):\Delta_1\times \{0,1\}\to \R.
\end{equation}
These are determined by the two functions $S(\cdot,0):\Delta_1\to \R, S(\cdot,1):\Delta_1\to \R$,  \cite[$\S$ 3.1]{scoring_rule}. Furthermore we say $S(\cdot,\cdot)$ is continuous, if both $S(\cdot,0)$ and $S(\cdot,1)$ are.
The average over the scoring rule $S(\cdot,\cdot)$ is given as
\begin{equation}\label{gen_scoring}
Av_{S}((\hat{p_i},y_i)_{i\in[n]})=\frac{1}{n}\sum_{i\in[n]}S(\hat{p_i},y_i).
\end{equation}

\color{black} As a generalisation of Tjur's $R^2$ \eqref{Tjur}, we consider a generalised evaluation over a scoring rule $S(\cdot,\cdot)$ as 
\begin{equation}\label{gen_Tjur}
Ev_{S}((\hat{p_i},y_i)_{i\in[n]})=\frac{1}{|\cX_1|}\sum_{i\in\cX_1}S(\hat{p_i},1)-\frac{1}{|\cX_0|}\sum_{i\in\cX_0}S(\hat{p_i},0).
\end{equation}
This generalises Tjur's $R^2$ of \eqref{Tjur} as we obtain \eqref{Tjur} if $S(\hat{p},0)=S(\hat{p},1)=\hat{p}$ in \eqref{gen_Tjur}.

Our main question of interest is the investigation of a potential equivalence between 
\begin{enumerate}
\item
the Youden index for $2\times 2$ contingency tables and
\item a probability prediction obtained trough a logistic regression with the state of the classifier evaluated via either average real scoring rule \eqref{gen_scoring} or Tjur $R^2$-like metric  \eqref{gen_Tjur} for a real continuous scoring rule.
\end{enumerate}
 We will show that such a correspondences cannot hold, even when we take a difference between the average predictions over 0s and 1s measured with a continuous scoring function from \eqref{gen_Tjur}.
\subsection*{Acknowledgments}
We thank Lucia de Andres Bragado for helpful discussions.
\section{Notation and setting}
Consider data $\cX=\{(x_i,y_i)| \ i\in 1,\cdots, n\}$ where the $x_i$ are the predictors and $y_i$ the binary outcomes. Assume we have a (deterministic) classifier giving for each observation $i$ a value $f_c(x_i)=\hat{y_i}\in\{0,1\}$, as well as a predicted probability $f_p(x_i)=\hat{p_i}\in\Delta_1$. 

Of $n$ observations, we collect a summary in a \( 2 \times 2 \) contingency table with the numbers of when the classifier $f_c$ predicted a $0/1$ as well as when it was right and wrong:
\[
\begin{array}{c|c|c}
  & Y = 1 & Y = 0 \\
\hline
f_c(X) = 1 & a & b \\
f_c(X) = 0 & c & d \\
\end{array}
\]
The Youden index can be rewritten as
\begin{equation}\label{Youden2}
Youden(a,b,c,d)=\frac{a}{a+c}+\frac{d}{b+d}-1
\end{equation}

Consider a logistic regression for the data $\cX=\{(f_c(x_i),y_i)| \ i\in 1,\cdots, n\}$ with independent variable the classification result and dependent variable the true outcome. This logistic regression with maximum-likelihood estimate will hence give the following prediction probabilities 
$$
g_p(x_i):=\begin{cases}
\frac{a}{a+b}, \text{ if } f_c(x_i)=1\\
\frac{c}{c+d}, \text{ if } f_c(x_i)=0.
\end{cases}
$$
We are interested in the question of whether there exists a real scoring rule such that an equality like
\begin{equation}\label{equa_youden_av}
Youden=Av_{S}((\hat{p_i},y_i)_{i\in[n]})
\end{equation}
or
\begin{equation}\label{equa_youden_r2}
Youden=Ev_{S}((\hat{p_i},y_i)_{i\in[n]})
\end{equation}
holds. We will show impossibility results for both, hence that such a real scoring rule cannot exist. 
\begin{remark}\label{eq_cont}
In fact, in the case of \eqref{equa_youden_r2}, we will show this for the space of continuous function $\Delta_1\to\R$. As there is a homeomorphism $\Delta_1\simeq[0,1]$, it is sufficient to show this for all continuous functions $[0,1]\to\R$, whose set we denote by $\cC([0,1],\R)$. $\cC([0,1],\R)$  is an infinite dimensional vector space containing the identity function which inserted in \eqref{gen_Tjur} for $S(\cdot,0)$ and $S(\cdot,0)$ gives Tjur's $R^2$. However, it also contains, e.g., any univariate polynomial of an arbitrary degree.
\end{remark}
Before going further, we remark the following.
\begin{remark}
The expressions of Youden index \eqref{Youden} is only defined for $a+c\neq 0$ and $b+d\neq0$. We restrict to this in the following, and denote the set of elements we exclude by $\cC$ and the one we consider by $\cA$, i.e.,
$$\cC:=\{(a,b,c,d)\in\N_{\geq0}^4|a+c= \text{ or }b+d=0\},
$$
$$ \cA=\N_{\geq0}^4\setminus \cC.
$$
\end{remark}
To be more concrete, we write out the expression of $Av_{S}((g_p(x_i),y_i)_{i\in[n]})$ and $Ev_{S}((g_p(x_i),y_i)_{i\in[n]})$, where the first is given by
\begin{equation}\label{gen_score_applied}
Av_{S}((g_p(x_i),y_i)_{i\in[n]})=\frac{1}{a+b+c+d}\left(a\cdot S(\frac{a}{a+b},1)+c\cdot S(\frac{c}{c+d},0)+b\cdot S(\frac{a}{a+b},0)+d\cdot S(\frac{c}{c+d},1)\right)
\end{equation}
and the second by
\begin{equation}\label{gen_Tjur_applied}
Ev_{S}((g_p(x_i),y_i)_{i\in[n]})=\frac{1}{a+c}\left(a\cdot S(\frac{a}{a+b},1)+c\cdot S(\frac{c}{c+d},0)\right)-\frac{1}{b+d}\left(b\cdot S(\frac{a}{a+b},0)+d\cdot S(\frac{c}{c+d},1)\right)
\end{equation}
The question we want to answer is the following. Is there a real-valued continuous scoring function $S(\cdot,\cdot)$ such that either
\begin{itemize}
\item the evaluation in \eqref{gen_scoring} equals the Youden index \eqref{Youden2}, or  
\item the evaluation in \eqref{gen_Tjur_applied} equals the Youden index \eqref{Youden2}, 
\end{itemize}
for all $(a,b,c,d)\in\cA$?
\section{Results}
\begin{theorem}
There is no real-valued scoring rule $S(\cdot,\cdot)$ for binary outcomes such that its average in \eqref{gen_score_applied} equals the Youden index \eqref{Youden2} for all $(a,b,c,d)\in\cA$.
\end{theorem}
\begin{proof}
We prove it by contradiction. Hence assume the following equation holds for $(a,b,c,d)\in\cA$
\begin{equation}\label{equ_contr}
\frac{a}{a+c}+\frac{d}{b+d}-1=\frac{1}{a+b+c+d}\left(a\cdot S(\frac{a}{a+b},1)+c\cdot S(\frac{c}{c+d},0)+b\cdot S(\frac{a}{a+b},0)+d\cdot S(\frac{c}{c+d},1)\right)
\end{equation}
for some scoring function $S(\cdot,\cdot)$.
We next go through some particular cases of \eqref{equ_contr} to derive a contradiction.\\

$\underline{\textbf{a=b=0}}$:\\
 Then for all $c,d\in\N_{\geq1}$, we get the equation
$$0=cS(\frac{c}{c+d},0)+dS(\frac{c}{c+d},1),$$
which is equivalent to 
\begin{equation}\label{equ_contr1}
S(\frac{c}{c+d},0)=-\frac{d}{c}S(\frac{c}{c+d},1).
\end{equation}
$\underline{\textbf{c=d=0}}$:\\
 Then for all $a,b\in\N_{\geq1}$, we get the equation
$$0=a\cdot S(\frac{a}{a+b},1)+b\cdot S(\frac{a}{a+b},0),$$
which is equivalent to 
\begin{equation}\label{equ_contr2}
S(\frac{a}{a+b},0)=-\frac{a}{b}S(\frac{a}{a+b},1).
\end{equation}
Setting \eqref{equ_contr1} and \eqref{equ_contr2} equal, we obtain
$$\frac{a}{b}S(\frac{a}{a+b},1)=\frac{b}{a}S(\frac{a}{a+b},1),
$$
hence $S(\frac{a}{a+b},1)=0$ which overall is a contradiction as the Youden index can be nonzero.
\end{proof}

\begin{theorem}
There is no continuous real-valued scoring rule $S(\cdot,\cdot)$ for binary outcomes such that \eqref{gen_Tjur_applied} equals the Youden index \eqref{Youden2} for all $(a,b,c,d)\in\cA$.
\end{theorem}
\begin{proof}
We prove it by contradiction. Hence assume the following equation holds for $(a,b,c,d)\in\cA$
\begin{equation}\label{equ_contr}
\frac{a}{a+c}+\frac{d}{b+d}-1=\frac{1}{a+c}\left(a\cdot S(\frac{a}{a+b},1)+c\cdot S(\frac{c}{c+d},0)\right)-\frac{1}{b+d}\left(b\cdot S(\frac{a}{a+b},0)+d\cdot S(\frac{c}{c+d},1)\right)
\end{equation}
for some continuous scoring function $S(\cdot,\cdot)$. Then $S(\cdot,0),S(\cdot,1)$ are continuous functions w.l.o.g. from $[0,1]\to \R$ (see Remark \ref{eq_cont}).
We next go through some particular cases of \eqref{equ_contr} to derive a contradiction.\\

$\underline{\textbf{a=b=0}}$:\\
 Then for all $c,d\in\N_{\geq1}$, we get the equation
$$0=S(\frac{c}{c+d},0)-S(\frac{c}{c+d},1),$$
which is equivalent to 
$$S(\frac{c}{c+d},0)=S(\frac{c}{c+d},1).$$
By the assumed continuity of $S(\cdot,0),S(\cdot,1)$, this equality on $(0,1)\cap\mathbb{Q}$ extends to an equality in $[0,1]$ as $(0,1)\cap\mathbb{Q}$ is dense in $[0,1]$ \cite{rudin1976principles}. Hence the following holds for all
$z\in[0,1]$:
\begin{equation}\label{symm}
S(z,0)=S(z,1).
\end{equation}

$\underline{ \textbf{b=c=0}}$: \\
For all $a,d \in\N_{\geq1}$ we get the equation
\begin{equation}\label{symm2}
1=S(1,1)-S(0,1).
\end{equation}

 $\underline{ \textbf{a=d=0}}$: \\
For all $b,c \in\N_{\geq1}$ we get the equation
$$-1=S(1,0)-S(0,0).$$
By \eqref{symm} the above equation implies the following 
$$-1=S(1,1)-S(0,1),$$
which is a contradiction to \eqref{symm2}.

\end{proof}

\section{Discussion}

We have shown that there does not exist a scoring rule such that its average when applied to predicted probabilities considered equals the Youden index. Also there does not exist a continuous scoring rule such that a Tjur $ R^2 $-like metric for probabilities of the form considered equals the Youden index. However, in a broader context, the question remains open as to whether classification evaluation measures \cite{classification_powers,hand_classification} can equal probabilistic prediction evaluation measures \cite{harrell2015regression} when applied to the maximum likelihood estimate, where the classification result serves as independent variable and the true outcome as dependent variable.  

Potential evaluation measures of interest include pseudo-$ R^2 $ metrics, or other widely used evaluation functions \cite{Tjur,harrell2015regression}. Establishing such a correspondence could provide new insights into classification evaluation measures and inspire the development of novel approaches in this area. Since our result only establishes the impossibility of equating the Youden index with a average score or a Tjur $ R^2$-like metric, it represents just a first step in exploring these relationships.

 \bibliographystyle{plain}

 \bibliography{references} 
 
\end{document}